\documentclass[]{jot}
\usepackage[T1]{fontenc}

\usepackage{xcolor}
\usepackage{mathpartir}
\usepackage{amssymb}
\usepackage{amsthm}
\usepackage{amsmath}
\usepackage{stmaryrd}
\usepackage{dashbox}
\usepackage{subcaption}
\usepackage{tikz}

\usepackage{etoolbox}
\AtBeginEnvironment{array}{}

\definecolor{Blue}{rgb}{0.0,0.0,1}
\definecolor{Purple}{rgb}{0.4,0,0.6}

\newcommand{\annot}[1]{{\color{Blue}#1}}
\newcommand{\auxcode}[1]{{\color{Purple}#1}}

\newcommand{\llbrace}{\{\hspace{-3pt}[}
\newcommand{\rrbrace}{]\hspace{-3pt}\}}

\newcommand{\po}{\mathsf{po}}
\newcommand{\mo}{\mathsf{mo}}
\newcommand{\rf}{\mathsf{rf}}
\newcommand{\hb}{\mathsf{hb}}
\newcommand{\rpo}{\mathsf{rpo}}
\newcommand{\porf}{\mathsf{porf}}

\colorlet{colorPO}{gray!60!black}
\colorlet{colorCF}{red!60!black}
\colorlet{colorECF}{red!60!black}
\colorlet{colorJF}{blue!60!black}
\colorlet{colorRF}{green!60!black}
\colorlet{colorEW}{brown}
\colorlet{colorMO}{orange}
\colorlet{colorFR}{purple}
\colorlet{colorECO}{red!80!black}
\colorlet{colorSYN}{green!40!black}
\colorlet{colorHB}{blue}
\colorlet{colorPPO}{magenta}
\colorlet{colorPB}{olive}
\colorlet{colorSBRF}{olive}
\colorlet{colorRMW}{olive!70!black}
\colorlet{colorRS}{blue}
\colorlet{colorRELEASE}{blue!70!black}
\colorlet{colorSC}{olive!40!black}
\colorlet{colorPSC}{olive!40!black}
\colorlet{colorREL}{olive}
\colorlet{colorCONFLICT}{olive}
\colorlet{colorRACE}{olive}
\colorlet{colorWB}{orange!70!black}
\colorlet{colorSCB}{violet}
\colorlet{colorDETOUR}{teal}
\colorlet{colorDEPS}{violet}
\colorlet{colorFENCE}{olive}
\colorlet{colorCOV}{magenta!20}
\colorlet{colorISS}{blue!10!white}
\colorlet{colorVF}{purple!70!black}

\tikzset{
   every path/.style={>=stealth},
   po/.style={->,color=colorPO,shorten >=-0.5mm,shorten <=-0.5mm},
   sw/.style={->,color=colorSYN,shorten >=-0.5mm,shorten <=-0.5mm},
   sc/.style={->,color=colorSC,dotted,thick,shorten >=-0.5mm,shorten <=-0.5mm},
   rf/.style={->,color=colorRF,dashed,,shorten >=-0.5mm,shorten <=-0.5mm},
   hb/.style={->,color=colorHB,thick,shorten >=-0.5mm,shorten <=-0.5mm},
   mo/.style={->,color=colorMO,dotted,very thick,shorten >=-0.5mm,shorten <=-0.5mm},
   co/.style={->,color=colorMO,dotted,thick,shorten >=-0.5mm,shorten <=-0.5mm},
   no/.style={->,dotted,thick,shorten >=-0.5mm,shorten <=-0.5mm},
   fr/.style={->,color=colorFR,dotted,thick,shorten >=-0.5mm,shorten <=-0.5mm},
   deps/.style={->,color=colorDEPS,dotted,thick,shorten >=-0.5mm,shorten <=-0.5mm},
   ppo/.style={->,color=colorPPO,shorten >=-0.5mm,shorten <=-0.5mm},
   rmw/.style={->,color=colorRMW,thick,shorten >=-0.5mm,shorten <=-0.5mm},
   detour/.style={->,color=colorDETOUR,shorten >=-0.5mm,shorten <=-0.5mm},
   cf/.style={-,snake=zigzag,segment amplitude=1pt,segment length=3pt,colorCF},
   ew/.style={<->,dashed,,shorten >=-0.5mm,shorten <=-0.5mm,color=colorEW},
   jf/.style={->,color=colorJF,dotted,thick,shorten >=-0.5mm,shorten <=-0.5mm},
   vf/.style={->,color=colorVF,dashed,shorten >=-0.5mm,shorten <=-0.5mm},
}

\newcommand{\Init}{\mathsf{Init}}
\newcommand{\lR}{{\mathtt{R}}}
\newcommand{\lW}{{\mathtt{W}}}
\newcommand{\lPO}{{\color{colorPO}\mathtt{po}}}
\newcommand{\lRF}{{\color{colorRF} \mathtt{rf}}}

\newtheorem{definition}{Definition}
\newtheorem{lemma}{Lemma}
\newtheorem{theorem}{Theorem}
\newtheorem{example}{Example}

\jotdetails{
    volume=25,      
    number=3,       
    articleno=a5,   
    year=2025,      
    license=ccby    
}
\articletype{regular} 

\title{An approach for modularly verifying the core of Rust's atomic reference counting algorithm against the (Y)C20 memory consistency model}

\author{Bart Jacobs
}
\author{Justus Fasse
}

\affil{KU Leuven, Department of Computer Science, DistriNet Research Group, Leuven, Belgium}


\runningtitle{Verifying Core ARC against (Y)C20}

\runningauthor{Jacobs and Fasse}

\begin{abstract}

We propose an approach for modular verification of programs that use relaxed-consistency atomic
memory access primitives and fences. The approach is sufficient for verifying the core of Rust's Atomic
Reference Counting (ARC) algorithm. We first argue its soundness, when combined with a simple static analysis and admitting an open sub-problem,
with respect to
the C20 memory consistency model. We then argue its soundness, even in the absence of any static analysis and without any assumptions, with respect to YC20, a minor strengthening of XC20, itself a
recently proposed minor strengthening
of C20 that rules out out-of-thin-air behaviors but
allows load buffering. In contrast to existing work on verifying
ARC, we do not assume acyclicity of the union of the program-order and reads-from relations.
We define an interleaving operational semantics, prove
its soundness with respect to (Y)C20's axiomatic semantics, and then apply any existing program logic for
fine-grained interleaving concurrency, such as Iris.
\end{abstract}

\acknowledgement{We thank Evgenii Moiseenko for helpful interactions. This research is partially funded by the Research Fund KU Leuven, and by the Cybersecurity Research Program Flanders.}

\begin{document}
\maketitle
\urlstyle{rm}

\section{Introduction}

Most work on modular verification of shared-memory multithreaded programs so far (e.g.~\cite{iris,iris-groundup}) has assumed \emph{sequential consistency}, i.e.~that in each execution of the program, there exists some total order on the memory accesses such that each read access yields the value written by the most recent preceding write access in this total order. However, for performance reasons, many real-world concurrent algorithms, such as Rust's Atomic Reference Counting (ARC) algorithm\footnote{\url{https://doc.rust-lang.org/std/sync/struct.Arc.html}}, use \emph{relaxed-consistency} memory accesses that do not respect such a total order.

\begin{figure*}
\begin{subfigure}[b]{.80\textwidth}
$$\begin{array}{c | c | c}
& &\\
\begin{array}{c}
X = Y = 0\\
\begin{array}{l || l}
a = X; & \,b = Y;\\
\mathbf{if}(a == 1)\, & \,\mathbf{if}(b == 1)\\
\quad Y = 1; & \quad X = 1;
\end{array}\\
a = b = 1\\
(\textrm{LBD})
\end{array}
& 
\begin{array}{c}
X = Y = 0\\
\begin{array}{l || l}
a = X; & \,b = Y;\\
Y = 1; & \,X = 1;
\end{array}\\
a = b = 1\\
\\
(\textrm{LB})
\end{array}
& 
\begin{array}{c}
X = Y = 0\\
\begin{array}{l || l}
a = X; & \,b = Y;\\
\mathbf{fence}_\mathsf{acq};\, & \,\mathbf{fence}_\mathsf{acq};\\
Y = 1; & \,X = 1;
\end{array}\\
a = b = 1\\
(\textrm{LBf})
\end{array}\\
& &
\end{array}$$
\caption{Three memory consistency litmus tests}\label{fig:litmus-tests}
\end{subfigure}\hfill
\begin{subfigure}[b]{.16\textwidth}\centering
\begin{tikzpicture}[yscale=0.7,xscale=1]

  \node (init) at (1,  1.5) {$\Init$};

  \node (i11) at ( 0,  0) {$\lR(X, 1)$};
  \node (i12) at ( 0, -2) {$\lW(Y, 1)$};

  \node (i21) at ( 2,  0) {$\lR(Y, 1)$};
  \node (i22) at ( 2, -2) {$\lW(X, 1)$};

  \draw[po] (i11) edge node[left] {\small$\lPO$} (i12);
  \draw[po] (i21) edge node[left] {\small$\lPO$} (i22);

  \draw[rf] (i22) edge node[below] {\small$\lRF$} (i11);
  \draw[rf] (i12) edge node[below] {} (i21);

  \draw[po] (init) edge node[left] {\small$\lPO$} (i11);
  \draw[po] (init) edge node[right] {\small$\lPO$} (i21);

\end{tikzpicture}
\caption{A C20 execution graph}\label{fig:graph}
\end{subfigure}
\caption{Three memory consistency litmus tests and a C20 execution graph explaining the LB/LBD behavior}
\end{figure*}

Consider the three \emph{litmus tests} LBD, LB, and LBf in Fig.~\ref{fig:litmus-tests}, small concurrent programs accompanied by a precondition $X = Y = 0$ and a postcondition $a = b = 1$. We say the behavior specified by a litmus test is \emph{observable} if the program has an execution that satisfies the pre- and postcondition. When implemented using relaxed accesses, the LB (\emph{load buffering}) behavior is observable, because of instruction reorderings performed by the compiler and/or the processor. Since acquire fences have no effect in the absence of release operations, so is the LBf behavior. The LBD behavior, however, known as an \emph{out-of-thin-air (OOTA)} behavior, is not observable on any real execution platform. Recent editions of the C standard attempt to precisely describe the concurrency behaviors that a C implementation is allowed to exhibit, by listing a number of \emph{axioms} that each execution must satisfy. We refer to the formalization thereof used in \cite{xmm} as C20. These axioms allow the LB and LBf behaviors; unfortunately, since the C20 concept of an execution does not consider dependencies, it cannot distinguish LBD from LB and allows it too. Strengthening C20 to rule out OOTA behaviors like LBD while still allowing all desirable optimizations, such as the reorderings exhibited by LB, has been a long-standing open problem. In the meantime, most work on modular verification of relaxed-consistency programs \cite{rsl,fsl,fslpp,relaxed-rustbelt}, has assumed the absence of cycles in the union of program order ($\po$, the total order on a thread's memory accesses induced by the program's control flow) and the reads-from ($\rf$) relation relating each write event to the read events that read from it.
Without such cycles, it is possible to prove the soundness of program logics by induction on the size of $(\po\cup\rf)^+$ prefixes of the execution.
In addition to ruling out LBD, however, this acyclicity assumption also rules out LB and LBf.

A breakthrough was achieved, however, with the very recent proposal of XMM \cite{xmm}, a framework for concurrency semantics based on \emph{re-execution}. Given an underlying memory model, XMM slightly restricts it to rule out OOTA behaviors, without ruling out $\porf$ cycles altogether. Specifically, XMM incrementally builds executions through Execute steps and Re-Execute steps. An Execute step adds an event that reads from an existing event. Thus, it never introduces a $\porf$ cycle. A Re-Execute step fixes an $\rf^{-1}$-closed subset of the events of the original execution, called the \emph{committed set}, and then re-builds an execution from scratch, where read events may temporarily ``read from nowhere'' \emph{but only if they are in the committed set}. Re-Execute steps enable the construction of $\porf$ cycles, but the values read are \emph{grounded} by an earlier execution. Applying XMM to C20 yields XC20, which allows LB (and LBf) but rules out LBD.

In this paper, we propose an approach for modularly verifying relaxed-consistency programs that is sound in the presence of $\porf$ cycles. We use the core of ARC as a motivating example; its code and intended separation logic specification are shown in Fig.~\ref{fig:arc}.
Function $\mathsf{alloc}(v)$ allocates an ARC instance whose payload is $v$. Every owner of a permission to access the ARC instance at address $a$ holding payload $v$ (denoted $\mathsf{arc}(a, v)$) can read the payload using function $\mathsf{get}(a)$, implemented using a simple nonatomic read. Function $\mathsf{clone}(a)$ duplicates the permission to access $a$, and $\mathsf{drop}(a)$ destroys it. When the last permission is destroyed, the instance is deallocated. The first field of the ARC data structure is a counter, initialized to 1, incremented at each $\mathsf{clone}$ using a relaxed-consistency ($\mathbf{rlx}$) fetch-and-add (FAA) instruction, and decremented at each $\mathsf{drop}$.
Since the decrements have release ($\mathbf{rel}$) consistency, they all synchronize with the acquire ($\mathbf{acq}$) fence in the thread that reads counter value 1, ensuring that the deallocation does not race with any of the accesses.
\begin{figure}
$$\begin{array}{l l}
\begin{array}{l}
\mathsf{alloc}(v) =\\
\quad \annot{\{\mathsf{True}\}}\\
\quad \mathbf{cons}(1, v)\\
\quad \annot{\{r.\;\mathsf{arc}(r, v)\}}\\
\\
\mathsf{get}(a) =\\
\quad \annot{\{\mathsf{arc}(a, v)\}}\\
\quad [a + 1]_\mathbf{na}\\
\quad \annot{\{r.\;\mathsf{arc}(a, v) \land r = v\}}\\
\\
\mathsf{clone}(a) =\\
\quad \annot{\{\mathsf{arc}(a, v)\}}\\
\quad \mathbf{FAA}_\mathbf{rlx}(a, 1)\\
\quad \annot{\{\mathsf{arc}(a, v) * \mathsf{arc}(a, v)\}}\\
\end{array}
&
\begin{array}{l}
\mathsf{drop}(a) =\\
\quad \annot{\{\mathsf{arc}(a, v)\}}\\
\quad \mathbf{let}\ n = \mathbf{FAA}_\mathbf{rel}(a, -1)\ \mathbf{in}\\
\quad \mathbf{if}\ n = 1\ \mathbf{then}\ (\\
\quad\quad \mathbf{fence}_\mathbf{acq};\\
\quad\quad \mathbf{free}(a); \mathbf{free}(a + 1)\\
\quad )\\
\quad \annot{\{\mathsf{True}\}}\\
\end{array}
\end{array}$$
\caption{Core ARC, with desired specs}\label{fig:arc}
\end{figure}

As we will show, this algorithm is correct under C20 \emph{provided that the following two properties hold}:
\begin{eqnarray}
\begin{array}{l r}
\multicolumn{2}{l}{\textit{All accesses of the counter}\quad}\\
\multicolumn{2}{r}{\quad\textit{are either increments or release decrements.}}
\end{array}\label{eq:arc-accesses}\\
\textit{No access of the counter reads a value $v \le 0$.}\label{eq:arc-no-zero}
\end{eqnarray}
In particular, clearly, the presence of relaxed decrements would break the algorithm. Verifying (\ref{eq:arc-accesses}) or (\ref{eq:arc-no-zero}) using separation logic is tricky, and we don't know how to do it. Indeed, consider LBf above. Suppose that we want to verify that all modifications of $X$ and $Y$ are release writes, so that they synchronize with the acquire fences. The problem here is that both writes are \emph{behind} fences. To reason about the fences, we would have to assume the property before verifying it, which is unsound.
Instead, we simply assume these properties. Actually, (\ref{eq:arc-accesses}) can be verified using a \emph{syntactic} approach. By using a Java-like static type system to protect the $\mathsf{Arc.counter}$ field it is then possible to syntactically check all accesses to it. For (\ref{eq:arc-no-zero}), in contrast, we have no suggestions; we leave it as an open problem.

In \S\ref{sec:c20}, we briefly recall the C20 memory consistency model. In \S\ref{sec:c20-logic}, we propose an approach for verifying programs under C20, under certain assumptions to be checked by other means, and we use it to verify Core ARC. Specifically, we propose an operational semantics (opsem) for C20 programs instrumented with an \emph{atomic specification} for each location accessed atomically, that specifies a set of \emph{enabled operations} for the location, as well as a \emph{tied precondition} for each enabled operation and a \emph{tied postcondition} for each enabled operation-result pair, both elements of an algebra of \emph{tied resources}. We assume all access events that occur are enabled. In the opsem, an atomic access nondeterministically yields any result enabled under its atomic specification. We argue the opsem's soundness with respect to C20's axiomatic semantics, and we apply the Iris \cite{iris,iris-groundup} logic to the opsem to verify Core ARC.

In \S\ref{sec:yc20-logic}, then, we adapt this approach to obtain an approach for verifying programs under a minor strengthening of XC20, which we call YC20.
Without any assumptions, we verify Core ARC against YC20.
We use the grounding guarantees offered by YC20 to prove that all atomic accesses of a location are enabled under its atomic specification. But first we recall the XC20 memory consistency model in \S\ref{sec:xc20} and define YC20. We finish by discussing related work (\S\ref{sec:related-work}) and offering a conclusion (\S\ref{sec:conclusion}).

\section{The C20 memory consistency model}\label{sec:c20}

We here briefly recall the C20 memory consistency model. For a gentler presentation, see \citet{batty-phd}.

In C20, the memory operations $o \in \mathcal{O} = \{\mathsf{R}_\mathsf{na},$ $\mathsf{R}_\mathsf{rlx},$ $\mathsf{R}_\mathsf{acq},$ $\mathsf{fence}_\mathsf{acq},$ $\mathsf{fence}_\mathsf{rel}\} \cup \{\mathsf{W}_\mathsf{na}(v),$ $\mathsf{W}_\mathsf{rlx}(v),$ $\mathsf{W}_\mathsf{rel}(v)$ $|\ v \in \mathbb{Z}\} \cup \{\mathsf{RMW}_\mathsf{rlx}(f),$ $\mathsf{RMW}_\mathsf{rel}(f),$ $\mathsf{RMW}_\mathsf{acq}(f),$ $\mathsf{RMW}_\mathsf{acqrel}(f)$ $|\ f \in \mathbb{Z} \rightarrow \mathbb{Z}\}$ are the nonatomic, relaxed, and acquire reads, acquire and release fences, nonatomic, relaxed and release writes of a value $v$, and relaxed, release, acquire, and acquire-release read-modify-write (RMW) operations that atomically read a value $v$ from a location and write value $f(v)$, for some function $f \in \mathbb{Z} \rightarrow \mathbb{Z}$. Each memory event is labeled by a tuple $(t, \ell, v, o)$ specifying the event's thread $t \in \mathit{ThreadIds}$, location $\ell \in \mathbb{Z}$, result value $v \in \mathbb{Z}$, and operation $o \in \mathcal{O}$.

The result value of a write is always 0. The result value of an RMW is the value that is read, before the modification is performed. It follows that an event labelled by operation $\mathsf{RMW}(f)$ and result value $v$ writes value $f(v)$.

There is a special \emph{initialization thread} $t_\mathsf{init} \in \mathit{ThreadIds}$ that performs a nonatomic write of value 0 to each memory location used by the program.

In C20, the set of behaviors of a program is given by its set of consistent C20 execution graphs. A C20 execution graph is a tuple $(E, \mathsf{lab}, \mathsf{po}, \mathsf{rf}, \mathsf{mo}, \mathsf{rmw})$, where $E$ is a set of events, $\mathsf{lab}$ is a function that maps each event $e \in E$ to its label, and the \emph{program order} $\mathsf{po}$, \emph{reads-from} relation $\mathsf{rf}$, \emph{modification order} $\mathsf{mo}$, and \emph{read-modify-write} relation $\mathsf{rmw}$ are subsets of $E \times E$. A C20 execution graph is \emph{well-formed} if all of the following conditions hold:
\begin{itemize}
\item Program order relates the initialization event for each location $\ell$ to each other access of that location, and otherwise relates two events only if they belong to the same thread. For each non-initialization thread $t$, program order totally orders the events of $t$.
\item The reads-from relation only relates write or RMW events to read or RMW events. It only relates events writing a value $v$ to location $\ell$ to events reading value $v$ from location $\ell$. It relates at most one write or RMW event to any given read or RMW event. If it relates events $e_1$ and $e_2$ by the same thread, $e_1$ precedes $e_2$ in program order.
\item For each location $\ell$, modification order totally orders the write and RMW events on $\ell$, and only relates writes or RMW events, and only relates events on the same location.
\item $\mathsf{rmw}$ only relates read events to write events. It only relates events to their immediate program order successor.
\end{itemize}

Unless otherwise noted, we only consider well-formed execution graphs.

We say a read or RMW $r$ reads from a write or RMW $w$ if $(w, r) \in \mathsf{rf}$. We say an execution graph is $\mathsf{rf}$-complete if each read and each RMW reads from some write or RMW.

This definition allows for two ways to represent RMWs: as singular events labelled by an RMW operation, and as pairs of a read and write event related by $\mathsf{rmw}$. We say an execution is \emph{high-level} if its $\mathsf{rmw}$ relation is empty, and \emph{low-level} if it has no events labelled by RMW operations. For each high-level execution, there is exactly one corresponding low-level execution (up to graph isomorphism), obtained by replacing each RMW event by the corresponding pair of read and write events and $\mathsf{rmw}$ edge. In the other direction, the correspondence is not unique since a pair of read and write events does not uniquely determine an RMW operation's update function $f$.

To define consistency of a C20 execution graph, we first need to define the derived relations $\mathsf{sw}$ (\emph{synchronizes-with}), $\mathsf{hb}$ (\emph{happens-before}), $\mathsf{fr}$ (\emph{from-reads}), and $\mathsf{eco}$ (\emph{extended coherence order}):
\begin{itemize}
\item $\mathsf{sw}$ relates a release write or RMW $w$ or a release fence succeeded in program order by a relaxed write or RMW $w$ to an acquire read or RMW $r$ or an acquire fence preceded in program order by a relaxed read or RMW $r$ if $r$ reads from $w$ or there exists a \emph{release sequence} from some RMW that reads from $w$ to some RMW that $r$ reads from. A release sequence is a sequence of RMWs where each next one reads from the preceding one.
\item $\mathsf{hb}$ is the transitive closure of the union of $\mathsf{po}$ and $\mathsf{sw}$.
\item $\mathsf{fr}$ relates each read or RMW event $r$ to each modification order successor of the write or RMW that $r$ reads from (if any).
\item $\mathsf{eco}$ is the transitive closure of the union of $\mathsf{rf}$, $\mathsf{mo}$, and $\mathsf{fr}$.
\end{itemize}

We say a C20 execution graph is \emph{consistent} if it is $\mathsf{rf}$-complete and both of the following properties hold:
\begin{itemize}
\item It respects \emph{coherence}, meaning that $\mathsf{hb}$ is irreflexive and, furthermore, if an event $e$ happens-before an event $e'$, $e'$ is not related to $e$ by $\mathsf{eco}$ (i.e.~$\mathsf{eco}$ is consistent with $\mathsf{hb}$).
\item It respects \emph{atomicity of RMWs}, meaning that there is no event $e$ such that $\mathsf{fr}$ relates the read event of an RMW to $e$ and $\mathsf{mo}$ relates $e$ to the write event of the same RMW.
\end{itemize}

In C20, the set of executions of a program is given by the set of consistent C20 execution graphs generated by the program. The execution graph in Fig.~\ref{fig:graph} is an execution of the LB and LBD programs. It has a $\mathsf{porf}$ cycle.

We say a C20 execution graph has a data race if it contains two events on the same location, at least one of which is a write or RMW, and at least one of which is nonatomic. In C, if any of a program's executions has a data race, the program is considered to have undefined behavior. The goal of our verification approach is to verify absence of data races.

\section{A verification approach for C20}\label{sec:c20-logic}

In this section, we propose an operational semantics for C20 (\S\ref{sec:c20-opsem}) and we apply it to verify Core ARC (\S\ref{sec:arc-c20}).

\subsection{An operational semantics for C20}\label{sec:c20-opsem}

We define an interleaving operational semantics for an instrumented version of our programming language.

\subsubsection{Instrumented programs}

We add two auxiliary commands to the syntax of the programming language: $\mathbf{begin\_atomic}(\ell, \Sigma)$ and
$\mathbf{end\_atomic}(\ell)$.
When turning a nonatomic memory location $\ell$ into an atomic location using the $\mathbf{begin\_atomic}$ command,
one has to specify an \emph{atomic specification} $\Sigma = (v_0, \mathcal{R}_\mathsf{G}, \mathcal{R}_\mathsf{L}, \rho_0, \mathsf{pre}, \mathsf{post})$ consisting of an initial value $v_0$, 
cancellative commutative monoids\footnote{A monoid is an algebra $(R, \cdot, \varepsilon)$ given by a set $R$ with an associative binary \emph{composition} operator $\cdot$ and a unit element $\varepsilon$. It is \emph{cancellative} if $v \cdot v_0 = v' \cdot v_0 \Rightarrow v = v'$.} of global tied resources $\mathcal{R}_\mathsf{G}$ (ranged over by $\rho$) and local tied resources $\mathcal{R}_\mathsf{L}$ (ranged over by $\theta$), an initial global tied resource $\rho_0 \in \mathcal{R}_\mathsf{G}$, a partial function $\mathsf{pre} : \mathcal{O} \rightharpoonup \mathcal{R}_\mathsf{G} \times \mathcal{R}_\mathsf{L}$, mapping each \emph{enabled operation} $o$ to its \emph{tied precondition} $(\rho, \theta)$ consisting of a \emph{global tied precondition} $\rho$ and a \emph{local tied precondition} $\theta$, and a partial function $\mathsf{post} : \mathcal{O} \times \mathbb{Z} \rightharpoonup \mathcal{R}_\mathsf{G} \times \mathcal{R}_\mathsf{L}$ mapping a pair of an enabled operation $o$ and a result value $v$ at which $o$ is \emph{enabled} to a \emph{tied postcondition} $(\rho', \theta')$ consisting of a \emph{global tied postcondition} $\rho'$ and a \emph{local tied postcondition} $\theta'$.\footnote{We consider only \emph{valid} atomic specifications, where no nonatomic operations or accesses are enabled.}

\begin{example}\label{ex:atomic-spec}
For the Core ARC proof, we will use the following atomic specification: $v_0 = 1$, $\mathcal{R}_\mathsf{G} = \mathcal{R}_\mathsf{L} = (\mathbb{N}, +, 0)$, $\rho_0 = 1$,
$\mathsf{pre} = \{(\mathsf{FAA}_\mathsf{rlx}(1), (1, 0)),$ $(\mathsf{FAA}_\mathsf{rel}(-1), (1, 0)), (\mathsf{fence}_\mathsf{acq}, (0, 1))\}$, and $$\begin{array}{l l l}
\mathsf{post} = & & \{((\mathsf{FAA}_\mathsf{rlx}(1), z), (2, 0))\ |\ z.\;1 \le z\}\\
& \cup\; & \{((\mathsf{FAA}_\mathsf{rel}(-1), z), (0, 0))\ |\ z.\;2 \le z\}\\
& \cup\; & \{((\mathsf{FAA}_\mathsf{rel}(-1), 1), (0, 1))\}\\
& \cup\; & \{((\mathsf{fence}_\mathsf{acq}, 0), (0, 1))\}
\end{array}$$
In words: we start out with one unit of global tied resource. The enabled operations are the relaxed increments, the release decrements, and the acquire fences. The enabled operation-result pairs are the relaxed increments that read a positive value, the release decrements that read a positive value, and the acquire fences (which always have result value 0). Incrementing and decrementing both consume one unit of global tied resource. Incrementing produces two units, and decrementing produces nothing, except if the result value (i.e.~the value that was read, before the decrement) is 1, in which case it produces one unit of \emph{local} tied resource.
Accesses are enabled only at result values (i.e.~values read, before the modification) $v \ge 1$. A fence consumes one unit of local tied resource and produces it again.

Here is a more readable Hoare triple-style representation:
\begin{mathpar}
\inferrule{1 \le v}{
\{(1, 0)\}\ \mathsf{FAA}_\mathsf{rlx}(1) \rightsquigarrow v\ \{(2, 0)\}
}
\and
\inferrule{2 \le v}{
\{(1, 0)\}\ \mathsf{FAA}_\mathsf{rel}(-1) \rightsquigarrow v\ \{(0, 0)\}
}
\and
\{(1, 0)\}\ \mathsf{FAA}_\mathsf{rel}(-1) \rightsquigarrow 1\ \{(0, 1)\}
\and
\{(0, 1)\}\ \mathsf{fence}_\mathsf{acq} \rightsquigarrow 0\ \{(0, 1)\}
\end{mathpar}
\end{example}

For now, we assume that for any $\mathbf{begin\_atomic}(\ell, \Sigma)$ event, all accesses $a$ of $\ell$ that do not happen before it happen after it and enabled under $\Sigma$, i.e.~$(\mathsf{lab}(a).o, \mathsf{lab}(a).v) \in \mathrm{dom}\,\mathsf{post}$.\footnote{This implies the program has no $\mathbf{end\_atomic}(\ell)$ operations, and also does not deallocate the location, suggesting a garbage-collected setting. Note: other resources, like the ARC's payload, might still be non-garbage-collected native resources.} In \S\ref{sec:yc20-logic}, we eliminate all of these assumptions.

The syntax of instrumented programs is as follows:
$$\begin{array}{r l}
e ::= & v\ |\ x\ |\ e + e\ |\ e = e\\
c ::= & \mathbf{cons}(\overline{e})\ |\ [e]_\mathbf{na}\ |\ [e] :=_\mathbf{na} e\ |\ \mathbf{free}(e)\\
& |\ \auxcode{\mathbf{begin\_atomic}(e, \Sigma)}\ |\ \auxcode{\mathbf{end\_atomic}(e)}\ |\ o(e)\\
& |\ \mathbf{if}\ e\ \mathbf{then}\ c\ |\ \mathbf{let}\ x = c\ \mathbf{in}\ c\ |\ \mathbf{fork}(c)
\end{array}$$
Command $o(e)$ applies memory access operation $o$ to the location yielded by expression $e$. Notice that even fence operations are qualified by a location this way. Indeed, our operational semantics takes into account only the synchronization induced by the interplay between the fence and the accesses of this particular location. It is future work to investigate the severity of this incompleteness and to lift it, if necessary.

\subsubsection{Consistency of a tied resource with an event}

We derive a cancellative commutative monoid of \emph{thread-bound tied resources} $\mathcal{R}_\mathsf{B} = \mathit{ThreadIds} \rightarrow \mathcal{R}_\mathsf{L}$, ranged over by $\Theta$. Each element of $\mathcal{R}_\mathsf{B}$ associates a local tied resource with each thread. Its composition is the pointwise composition and its unit element is the function that maps all thread ids to $\varepsilon$. We further derive a cancellative commutative monoid of \emph{total tied resources} $\mathcal{R}_\mathsf{T} = \mathcal{R}_\mathsf{G} \times \mathcal{R}_\mathsf{B}$ (ranged over by $\omega$), whose composition is the componentwise composition and whose unit element is $(\varepsilon, \varepsilon)$. We will abuse $\cdot$ and $\varepsilon$ to denote the composition and unit element of any of these monoids. Initially, there are no local tied resources, so the initial total tied resource is $(\rho_0, \varepsilon)$.

An \emph{atomic location trace} $(G, E_\mathsf{at}, \mathit{init})$ for a location $\ell$ under an atomic specification $\Sigma$ consists of a consistent C20 execution graph $G$, a set $E_\mathsf{at} \subseteq G.E$ of events on $\ell$, all of whose operations are enabled under $\Sigma$ (i.e.~$\mathsf{lab}(E_\mathsf{at}).o \subseteq \mathrm{dom}\,\Sigma.\mathsf{pre}$), and a nonatomic write $\mathit{init} \in G.E$ of $\Sigma.v_0$ to $\ell$ that happens-before all events in $E_\mathsf{at}$, such that each read or RMW event in $E_\mathsf{at}$ reads from some event in $E_\mathsf{at} \cup \{\mathit{init}\}$.

We say an event is \emph{enabled} under $\Sigma$ if the event is labeled by $(t, \ell, v, o)$ such that $o$ is enabled at result $v$: $(o, v) \in \mathrm{dom}\,\Sigma.\mathsf{post}$.

We say an atomic location trace is \emph{fully enabled} if all events in $E_\mathsf{at}$ are enabled.

We say a total tied resource $\omega$ is \emph{consistent} with an operation $o$ by a thread $t$ resulting in a value $v$ under an atomic specification $\Sigma$, denoted $\Sigma, t, v, o \vDash \omega$, if there exists a fully enabled atomic location trace under $\Sigma$ that includes an event $e \in E_\mathsf{at}$ labeled by an operation $o$ by thread $t$ resulting in $v$ and there exists a subset $E_\mathsf{ex} \subseteq E_\mathsf{at}$ that includes all events in $E_\mathsf{at}$ that happen before $e$ but does not include $e$ itself or any events that happen after $e$, such that for every total ordering of $E_\mathsf{ex}$ consistent with happens before, starting from $(\Sigma.\rho_0, \varepsilon)$ it is possible to, in this order, consume the tied precondition and produce the tied postcondition of each event $e' \in E_\mathsf{ex}$ and arrive at $\omega$. We say that the consistency is \emph{witnessed} by the atomic location trace, $e$, and $E_\mathsf{ex}$.

\subsubsection{Operational semantics}

We instrument the state space of the programming language as follows. The regular heap $h$ stores the values of nonatomic cells only; for
atomic cells the semantics does not track the value; it only tracks, in the \emph{atomic heap} $A$, each atomic location's atomic specification and total tied resource.

\begin{figure*}
\begin{mathpar}
\inferrule[BeginAtomic]{
(\ell, \Sigma.v_0) \in h
}{
(h, A, \mathbf{begin\_atomic}(\ell, \Sigma)) \stackrel{t}{\rightarrow}_\mathsf{h} (h[\ell:=\varnothing], A[\ell:=(\Sigma, (\Sigma.\rho_0, \varepsilon))], 0)
}
\and
\inferrule[EndAtomic]{
(\ell, (\Sigma, \omega)) \in A
}{
(h, A, \mathbf{end\_atomic}(\ell)) \stackrel{t}{\rightarrow}_\mathsf{h} (h[\ell:=v], A[\ell:=\bot], 0)
}
\and
\inferrule[AtomicOp]{
(\ell, (\Sigma, (\rho\cdot\rho', \Theta[t:=\theta\cdot\theta']))) \in A\\
(o, (\rho, \theta)) \in \Sigma.\mathsf{pre}\\
((o, v), (\rho'', \theta'')) \in \Sigma.\mathsf{post}\\
\Sigma, t, v, o \vDash (\rho\cdot\rho', \Theta[t:=\theta\cdot\theta'])
}{
(h, A, o(\ell)) \stackrel{t}{\rightarrow}_\mathsf{h}
(h, A[\ell := (\Sigma, (\rho''\cdot\rho', \Theta[t:=\theta''\cdot\theta']))], v)
}
\and
\inferrule[AtomicOp-Stutter]{
(\ell, (\Sigma, (\rho\cdot\rho', \Theta[t:=\theta\cdot\theta']))) \in A\\
(o, (\rho, \theta)) \in \Sigma.\mathsf{pre}
}{
(h, A, o(\ell)) \stackrel{t}{\rightarrow}_\mathsf{h}
(h, A, o(\ell))
}
\end{mathpar}
\caption{Selected step rules of the operational semantics}\label{fig:opsem}
\end{figure*}

A \emph{configuration} $\gamma = (h, A, T)$ consists of a regular heap $h$, an atomic heap $A$, and a \emph{thread pool} $T$ which is a partial function mapping the thread id of each started thread to the command it is currently executing. A configuration $\gamma$ can \emph{step} to another configuration $\gamma'$, denoted $\gamma \rightarrow \gamma'$, if there is a started thread $t$ whose command is of the form $K[c]$, where $K$ is a reduction context and $c$ is a command that can step in the current state according to the head step relation $\stackrel{t}{\rightarrow}_\mathsf{h}$, defined by a number of step rules, shown in Fig.~\ref{fig:opsem} and in the Appendix.

Notice that, per rules \textsc{AtomicOp} and \textsc{AtomicOp-Stutter}, if the tied precondition for an operation is not available, the command gets stuck.\footnote{If a configuration where some thread is stuck is reachable, the program is
considered unsafe. We prove below that if a program is safe, it has no undefined behavior per C20.}
Otherwise, the operation consumes the tied precondition and nondeterministically produces a result value $v$ at which the operation is enabled, and the corresponding global tied postcondition $\rho''$ and local tied postcondition $\theta''$.

An operation stutters while the total tied resources are not consistent with the operation. This means the thread might be blocked until all events that happen-before it in the C20 execution have happened in the opsem execution. Indeed, the set
$E_\mathsf{ex}$ in the definition of consistency reflects the set of events that have happened in the opsem execution, and the opsem nondeterministically executes the events in any order consistent with $\hb$. Importantly, however, this does \emph{not} imply that all (or any) of the operation's $\rf^+$-predecessors, which explain the operation's result, have necessarily happened in the opsem execution yet.

We also allow operations to stutter spuriously; this does not matter since our approach is for safety only, not termination or other liveness properties.

\begin{theorem}
If a program is safe under the operational semantics, then it has no undefined behavior under C20 semantics.
\end{theorem}
\begin{proof}
Fix a C20 execution. It is sufficient to prove, for each $\mathsf{hb}$-prefix of the execution, that the prefix is data-race-free and that the opsem configuration corresponding to the prefix is reachable by taking the opsem steps corresponding to the events of the prefix in any order consistent with $\hb$. By induction on the size of the prefix. \emph{For more details, see Appendix~\ref{app:c20-soundness}.}
\end{proof}

\subsection{Proof of Core ARC}\label{sec:arc-c20}

\begin{figure}
$$\begin{array}{l}
\mathbf{public}\ \mathbf{class}\ \mathsf{Arc}\langle\mathsf{T}\ \mathbf{extends}\ \mathsf{Closeable}\rangle\ \{\\
\quad \mathbf{private}\ \mathbf{final}\ \mathsf{AtomicLong}\ \mathsf{counter} = \mathbf{new}\ \mathsf{AtomicLong}(1);\\
\quad \mathbf{public}\ \mathbf{final}\ \mathsf{T}\ \mathsf{payload};\\
\\
\quad \mathbf{public}\ \mathsf{Arc}(\mathsf{T}\ \mathsf{payload})\\
\quad \{\ \mathbf{this}.\mathsf{payload} = \mathsf{payload};\ \}\\
\\
\quad \mathbf{public}\ \mathbf{void}\ \mathsf{clone}()\ \{\ \mathsf{counter}.\mathsf{fetch\_and\_add\_relaxed}(1);\ \}\\
\\
\quad \mathbf{public}\ \mathbf{void}\ \mathsf{drop}()\ \{\\
\quad\quad \mathbf{long}\ \mathsf{v} = \mathsf{counter}.\mathsf{fetch\_and\_add\_release}(-1);\\
\quad\quad \mathbf{if}\ (\mathsf{v}\;{=}{=}\;1)\ \{\\
\quad\quad\quad \mathsf{AtomicLong}.\mathsf{fence\_acquire}();\\
\quad\quad\quad \mathsf{payload}.\mathsf{close}();\\
\quad\quad \}\\
\quad \}\\
\}
\end{array}$$
\caption{Core ARC in a Java-like language with C20 memory semantics}\label{fig:arc-java}
\end{figure}

To satisfy the assumptions stated in \S\ref{sec:c20-opsem}, except for the assumption that no access of the counter reads a value $v \le 0$, we verify a version of the code in Fig.~\ref{fig:arc} written in a hypothetical Java-like language with C20 memory semantics, shown in Fig.~\ref{fig:arc-java}. We assume the hypothetical $\mathsf{AtomicLong}$ class guarantees its initialization happens-before all accesses. The goal is to prove that, if accessing the payload requires an $\mathsf{arc}$ permission, then all accesses happen-before the closing of the payload.

For the atomic location, we use the atomic specification of Example~\ref{ex:atomic-spec}.

Notice that the global tied resource seems to track the value of the location exactly. Since a location's tied resource is tracked in the opsem's atomic heap, does this not imply sequentially consistent semantics? No, because the tied resource reflects only the events that have happened so far in the opsem execution; it does not take into account the ``future'' events, even though the former may read from the latter. It follows that in general, the global tied resource at the time of an access does \emph{not} match the value read by that access.

As we will see below, each $\mathsf{arc}$ permission will include ownership of one unit of global tied resource.

This choice of atomic specification is motivated by the following lemmas, which guarantee that only one decrement reads 1, and that after the acquire fence, there are no outstanding $\mathsf{arc}$ permissions:
\begin{lemma}\label{lem:arc-decrement}
If a total tied resource $(\rho, \Theta)$ is consistent with a decrement that reads 1, then $\Theta = 0$\footnote{We abuse 0 to denote the thread-bound resource that maps each thread to 0.}:
$$\Sigma, t, 1, \mathsf{FAA}_\mathsf{rel}(-1) \vDash (1, 0) \cdot (\rho, \Theta) \Rightarrow \Theta = 0$$
\end{lemma}
\begin{proof}
Fix an atomic location trace, an event $e$, and a set $E_\mathsf{ex}$ that witnesses the consistency. By contradiction: assume some event $e'$ in $E_\mathsf{ex}$ has a nonzero local tied postcondition. It follows that $e'$ is a decrement that reads 1. However, since both $e$ and $e'$ are decrements that read 1, some increment that reads 0 must intervene between $e'$ and $e$ in modification order. This contradicts the fact that all events in the atomic location trace are enabled. 
\end{proof}

\begin{lemma}\label{lem:arc-fence}
If a total tied resource $(\rho, \Theta)$ is consistent with an acquire fence, then $\rho = 0$:
$$\Sigma, t, 0, \mathsf{fence}_\mathsf{acq} \vDash (0, 0[t:=1]) \cdot (\rho, \Theta) \Rightarrow \rho = 0$$
\end{lemma}
\begin{proof}
Fix an atomic location trace, an event $e$, and a set $E_\mathsf{ex}$ that witnesses the consistency. By the existence of the local tied resource in thread $t$, some event $e'$ must precede $e$ in thread $t$ that produces it. $e'$ must be a decrement that reads value 1. It follows that the events that $\mo$-precede $e'$, as ordered by $\mo$, must consist of the $\mathit{init}$ event, followed by $n$ increments and $n$ decrements, in some order. Since the decrements (including $e'$) are release events, they happen-before $e$ and are therefore in $E_\mathsf{ex}$. Since their tied preconditions can be consumed, there must be at least $n$ increments in $E_\mathsf{ex}$, each of which happens-before one of these decrements and therefore $\mo$-precedes $e'$. It suffices to prove that the total number of increments in $E_\mathsf{ex}$ is $n$. For this, in turn, it suffices to prove that all increments in $E_\mathsf{ex}$ $\mo$-precede $e'$. Suppose an increment $e'' \in E_\mathsf{ex}$ $\mo$-succeeds $e'$. It therefore cannot happen-before $e'$ or any of the decrements $\mo$-preceding $e'$. But then there exists an order of the events in $E_\mathsf{ex}$ consistent with $\hb$ where $e''$ occurs at a point where the global tied resource has already reached 0 so the tied precondition for $e''$ cannot be consumed, which is a contradiction.
\end{proof}

We apply stock Iris \cite{iris,iris-groundup} to our operational semantics. For an atomic location $\ell$ we use an Iris cancellable invariant with tag $\tau$ containing $\mathsf{Inv}_{\ell,\tau,\gamma,\gamma'}$ defined as follows, where $F$ ranges over bags of positive real numbers, $\sum F$ means the sum of the elements of $F$, $\#F$ means the number of elements in $F$, and $\ell \stackrel{q}{\mapsto} \_$ denotes a fractional permission \cite{fractions-boyland,permissions-parkinson} with fraction $q \in [0, 1] \subseteq \mathbb{R}$ for location $\ell$:
$$\begin{array}{l}
\mathsf{Inv}_{\ell,\tau,\gamma,\gamma'} = \exists \rho, \Theta, F.\\
\quad \ell \mapsto_\mathsf{at} (\Sigma, (\rho, \Theta)) * (\Theta \neq 0 \lor [\tau]_{1/2}) * \ell + 1 \stackrel{1 - \sum F}{\mapsto} \_\\
\quad *\; [\tau]_{(1 - \sum F)/2} * \dbox{\ensuremath{\bullet (\rho, \Theta)}}_\gamma * \dbox{\ensuremath{\bullet F}}_{\gamma'} \land \rho = \#F
\end{array}$$
The invariant holds full permission for the atomic location. Furthermore, one half of the cancellation token $[\tau]$ is initially inside the invariant and removed by the thread that decreases the count to zero, exploiting Lemma~\ref{lem:arc-decrement} (at which point a local tied resource is produced, which is never destroyed). There is an element $q \in F$ for each outstanding $\mathsf{arc}$ permission, denoting the fraction of the ownership of $\ell + 1$ owned by that $\mathsf{arc}$ permission; the remainder is inside the invariant. That $\mathsf{arc}$ permission also owns a fraction $q/2$ of the cancellation token. ``Fictional ownership'' by $\mathsf{arc}$ permissions of tied resources and elements of $F$ is realized using Iris \emph{ghost cells} $\gamma$ and $\gamma'$ using the \textsc{Auth} resource algebra, whose \emph{authoritative parts}, denoted by $\bullet$, are held in the invariant, and for which a \emph{fragment}, denoted by $\circ$, is held in each $\mathsf{arc}$ permission. A ghost cell's fragments always add up exactly to the authoritative part.

We define predicate $\mathsf{arc}$ as follows:
$$\begin{array}{l}
\mathsf{arc}(\ell, v) = \exists q, \tau, \gamma, \gamma'.\;[\tau]_{q/2} * {}^\tau\boxed{\mathsf{Inv}_{\ell,\tau,\gamma,\gamma'}} * \ell + 1 \stackrel{q}{\mapsto} v\\
\quad *\; \dbox{\ensuremath{\circ \llbrace q\rrbrace}}_{\gamma'} * \dbox{\ensuremath{\circ (1, 0)}}_{\gamma}
\end{array}$$

Immediately before the $\mathbf{fence}_\mathbf{acq}$ command, thread $t$ owns $[\tau]_{1/2} * \dbox{\ensuremath{\circ (0, 0[t:=1])}}_\gamma$. When the fence succeeds, thanks to Lemma~\ref{lem:arc-fence}, the thread can cancel the invariant, end the atomic location, and free the memory cells.

\section{The XC20 and YC20 memory consistency models}\label{sec:xc20}

We brief recall the XC20 memory consistency model \cite{xmm}, and we define our minor strengthening of it, which we call YC20.

Given a C20 execution graph $G$, we define the \emph{relaxed program order} $\mathsf{rpo}$ as the transitive closure of the subset of $\mathsf{po}$ that only relates events $e$ and $e'$ if:
\begin{itemize}
\item $e$ is a relaxed (or stronger\footnote{with $\mathsf{na} < \mathsf{rlx} < \mathsf{acq} < \mathsf{acqrel}$ and $\mathsf{rlx} < \mathsf{rel} < \mathsf{acqrel}$}) read or RMW and $e'$ is an acquire fence, or
\item $e$ is an acquire (or stronger) event, or
\item $e'$ is a release (or stronger) event, or
\item $e$ is a release fence and $e'$ is a relaxed (or stronger) write.
\end{itemize}

A program behavior is allowed by the XC20 memory consistency model if it matches a consistent C20 execution graph that can be \emph{constructed} from the empty graph through a number of \emph{XMM construction steps}. There are two kinds of such steps: Execute steps and Re-Execute steps.

An Execute step simply adds a $\mathsf{porf}$-maximal event: it relates graphs $G$ and $G'$ if both are consistent C20 execution graphs and $G'.E = G.E \cup \{e\}$ and $G'|_{G.E} = G$ and $e$ is $\mathsf{porf}$-maximal in $G'$.

A Re-Execute step selects an $\mathsf{rf}$-complete subset of the events of the original execution, called the \emph{committed} set, and a $\mathsf{po}$-maximal $\mathsf{po}$-prefix of the committed set called the \emph{determined set}, and then constructs a new execution, starting from the determined set, where, at every step, a $\mathsf{po}$-maximal event $e$ is added such that either $e$ is not a read or RMW event, or $e$ reads from an event that was added earlier, or $e$ is in the committed set. The newly constructed execution must in the end be a consistent C20 execution graph and must include the entire committed set, which implies that all events that were added from the committed set again read from the same event that they read from in the original execution.

More precisely, a Re-Execute step relates graphs $G$ and $G'$ if all of the following hold:
\begin{itemize}
\item $G$ and $G'$ are consistent C20 execution graphs,
\item there exists a set of events $C$, called the \emph{committed set}, that is a subset of $G.E$ and $G'.E$ such that $G|_C = G'|_C$ and every read or RMW event in $C$ reads from an event in $C$,
\item there exists a subset $D$ of $C$, called the \emph{determined set}, that is $G.\mathsf{po}$-prefix-closed, i.e.~any event that $G.\mathsf{po}$-precedes an event in $D$ is also in $D$,
\item $D$ is $G.\mathsf{po}$-maximal in $C$, i.e.~if an immediate $G.\mathsf{po}$-successor of an event in $D$ is in $C$ then it is in $D$,
\item if $G'.\mathsf{rpo}$ relates an event $e$ to a non-determined event $e'$, then $e$ is a determined event,
\item $G'$ can be constructed from $G|_D$ through a sequence of Guided Steps under $C$.
\end{itemize}

A Guided Step under $C$ relates graphs $G$ and $G'$ if $G'$ is obtained by adding a $\mathsf{po}$-maximal event to $G$ that either is not a read or RMW event, or reads from an event in $G$, or is in $C$.

The definitions above define XMM and XC20 if \emph{low-level} executions are used, where each RMW is represented as a pair of a read event and a write event, related by the $\mathsf{rmw}$ relation, and they define YMM and YC20 if \emph{high-level} executions are used, where each RMW is represented as a single event labelled by an RMW operation. The only difference is that in XMM, it is possible for only one of the two events constituting an RMW to be in the committed set of a Re-Execute step, whereas in YMM, an RMW is entirely in the committed set or not at all. In XMM, if a Re-Execute step relates graphs G and G', as far as we can tell, it is possible for the write of a relaxed increment in G to be used to ``justify'' a relaxed decrement in G', which would seem to break our approach for verifying Core ARC.

\citet{xmm} proved a number of important results about XC20, including optimal compilation schemes to the main instruction set architectures and the soundness of a number of important program transformations performed by optimizing compilers. We would expect these to still hold for YC20, but verifying this is future work.

\section{A verification approach for YC20}\label{sec:yc20-logic}
In this section, we propose an operational semantics for YC20 (\S\ref{sec:yc20-opsem}) and we apply it to verify Core ARC (\S\ref{sec:arc-yc20}).

\subsection{An operational semantics for YC20}\label{sec:yc20-opsem}

We say a total tied resource $\omega$ is \emph{grounding-consistent} with an operation $o$ by a thread $t$ resulting in a value $v$ under an atomic specification $\Sigma$, denoted $\Sigma, t, v, o \vDash_\mathsf{g} \omega$, if there exists an atomic location trace under $\Sigma$ where $E_\mathsf{at}$ includes an event $e$ labeled by an operation $o$ by thread $t$ resulting in $v$ and all events in $E_\mathsf{at} \setminus \{e\}$ are enabled under $\Sigma$ and there exists a subset $E_\mathsf{ex} \subseteq E_\mathsf{at}$ that includes all events that happen before $e$ as well as all release events in $E_\mathsf{at} \setminus \{e\}$ but does not include $e$ itself or any events that happen after $e$, such that for every total ordering of $E_\mathsf{ex}$ consistent with happens before, starting from $(\Sigma.\rho_0, \varepsilon)$ it is possible to, in this order, consume the tied precondition and produce the tied postcondition of each event $e' \in E_\mathsf{ex}$ and arrive at $\omega$. We say the grounding-consistency is \emph{witnessed} by the atomic location trace, $e$, and $E_\mathsf{ex}$.

Our operational semantics for YC20 is exactly the same as the one for C20 presented in \S\ref{sec:c20-opsem}, except that we do not make the assumptions made there, and that we add one constraint on atomic specifications: that each operation's tied precondition be \emph{sufficient}. Informally, this means that during the grounding process the tied precondition ensures the operation is performed only at a value at which it is enabled under the atomic specification. Formally, if $(o, (\rho, \theta)) \in \mathsf{pre}$ and $\Sigma, t, v, o \vDash_\mathsf{g} (\rho, \varepsilon[t:=\theta]) \cdot \omega$, then $(o, v) \in \mathrm{dom}\,\mathsf{post}$.

\begin{theorem}
If a program is safe under the operational semantics, then it has no undefined behavior under YC20 semantics.
\end{theorem}
\begin{proof}
It is sufficient to prove that each C20 execution reachable through YMM's Execute and Re-Execute steps is \emph{grounded} and has no undefined behavior. An execution is grounded if, essentially, each atomic event is enabled under its atomic specification. By induction on the number of YMM steps. First, we prove based on the definition of YMM that for any Execute or Re-Execute step, if the original execution is grounded, the new execution is \emph{weakly grounded}, meaning essentially that there exists some \emph{grounding order}, a total order on the events of the execution, such that each event is either grounded or reads from an event that precedes it in the grounding order. We then prove, for any weakly grounded C20 execution, that each $\hb$-prefix of each grounding-order-prefix-closed subset of the execution is data-race-free and grounded and corresponds to an opsem configuration that is reachable by taking the opsem steps corresponding to the events of the prefix in any order consistent with $\hb$. By induction on the size of the subset and nested induction on the size of the prefix. \emph{For more details, see Appendix~\ref{app:yc20-soundness}.}
\end{proof}

\subsection{Proof of Core ARC}\label{sec:arc-yc20}

We can now verify the unmodified ARC code from Fig.~\ref{fig:arc}, without any assumptions. The atomic specification used, the Lemmas~\ref{lem:arc-decrement} and \ref{lem:arc-fence} and the Iris proof are adopted unchanged from \S\ref{sec:arc-c20}. The only difference is that we now need to prove the following lemma:
\begin{lemma}
The tied preconditions are sufficient.
\end{lemma}
\begin{proof}
Fix an operation $o$ such that $(o, (\rho, \theta)) \in \mathsf{pre}$ and a value $v$ such that $\Sigma, t, v, o \vDash_\mathsf{g} (\rho, 0[t:=\theta]) \cdot \omega$. It suffices to prove that $(o, v) \in \mathrm{dom}\,\mathsf{post}$.
Let $(G, E_\mathsf{at}, \mathit{init})$ be an atomic location trace and $e \in E_\mathsf{at}$ and $E_\mathsf{ex} \subseteq E_\mathsf{at}$ an event and set of events that witness the grounding-consistency. The case $o = \mathsf{fence}_\mathsf{acq}$ is trivial because its only possible result value is 0; assume $o$ is an $\mathsf{FAA}$ operation. By $\rf$-completeness of the atomic location trace and the fact that all events other than $e$ are enabled under $\Sigma$, $e$ must read a nonnegative value $v$. It suffices to prove that $v \neq 0$. By contradiction; assume $v = 0$. Then there exists a sequence of $\mathsf{FAA}$ events in $E_\mathsf{at}$ such that the first one reads from $\mathsf{init}$, each next one reads from the previous one, and $e$ reads from the last one. It follows that this sequence contains $n$ increment operations and $n + 1$ decrement operations. Since $E_\mathsf{ex}$ contains all release events in $E_\mathsf{at} \setminus \{e\}$, it contains these $n + 1$ decrement operations. Since it is possible to consume all of their tied preconditions, $E_\mathsf{ex}$ must also contain $n$ increment operations that each happen-before one of these decrement operations. It follows that $E_\mathsf{ex}$ must contain the $n$ increment operations that $\mo$-precede $e$. It now suffices to prove that $E_\mathsf{ex}$ contains only these $n$ increment operations. Suppose there was some increment event $e' \in E_\mathsf{ex}$ that $\mo$-succeeds $e$. It follows that $e'$ cannot happen-before any of the decrements that $\mo$-precede $e$. It follows that there is an order on the events of $E_\mathsf{ex}$ consistent with $\hb$ where $e'$ occurs after exactly $n$ increments and $n + 1$ decrements so the global tied resource has already reached 0 at the point where the tied precondition for $e'$ is consumed, which is a contradiction.
\end{proof}

\section{Related work}\label{sec:related-work}

Core ARC was verified before in FSL++ \cite{fslpp}, an extension of Fenced Separation Logic \cite{fsl} with support for ghost state, as well as in the Iris-based approach of Relaxed RustBelt \cite{relaxed-rustbelt}. The latter even verified full ARC, except that they had to strengthen two relaxed accesses to acquire accesses. One of these was in fact a bug in ARC; the status of the other remains open. Both of these earlier proofs, however, assume the absence of $\porf$ cycles.

The concept of tied resources was inspired by AxSL \cite{axsl}, an Iris-based separation logic for the relaxed memory model of the Arm-A processor architecture, which allows load buffering. In AxSL, however, tied resources are tied to \emph{events}, not locations. While an in-depth comparison is future work, their logic and their soundness proof appear to be quite different from ours.

\section{Conclusion}\label{sec:conclusion}

We presented a preliminary result on the modular verification of relaxed-consistency programs under YC20 semantics, a minor strengthening of XC20, a new memory consistency model that
rules out out-of-thin-air behaviors while allowing load buffering, and we used it to verify Core ARC.
The most urgent next step is to further elaborate the soundness proof, and ideally mechanize it,
to mechanise the Core ARC proof, and to verify that the results proved for XC20 by \citet{xmm} also hold for YC20.
Others are to investigate whether the approach supports other important relaxed algorithms besides
Core ARC (such as full ARC), and how to apply the approach in a semi-automated verification tool such as VeriFast \cite{fvf}.
It would also be interesting to try and eliminate the open assumption we used to prove Core ARC under C20 semantics, or else to
obtain some kind of an impossibility result.

Our operational semantics is a non-intrusive extension of classical semantics for sequentially-consistent (SC) languages, such as
the HeapLang language targeted by Iris by default. This suggests its compatibility with existing applications, extensions, and tools
for logics for SC languages such as Iris and VeriFast. The opsem used by Relaxed RustBelt \cite{relaxed-rustbelt}, in contrast, replaces
the regular heap by a \emph{message pool} combined with three \emph{views} per thread (the acquire view, the current view, and the release view). Comparing the pros and cons of these approaches is important future work.

\bibliography{arc-paper}

\section*{About the authors}
\shortbio{Bart Jacobs}{is an associate professor at the DistriNet research group at the department of Computer Science at KU Leuven (Belgium). His main research interest is in modular formal verification of concurrent programs. \authorcontact[https://distrinet.cs.kuleuven.be/people/BartJacobs/]{bart.jacobs@kuleuven.be}}
\shortbio{Justus Fasse}{is a PhD student at the DistriNet research group at the department of Computer Science at KU Leuven (Belgium). His research is on modular formal verification of concurrent programs, with a focus on total correctness properties. \authorcontact[https://distrinet.cs.kuleuven.be/people/JustusFasse/]{justus.fasse@kuleuven.be}}

\setcounter{section}{0}
\renewcommand\thesection{\Alph{section}}

\section{Appendix}

\subsection{Operational semantics}

We define the \emph{reduction contexts} $K$ as follows:
$$K ::= \square\ |\ \mathbf{let}\ x = K\ \mathbf{in}\ c$$
We write $K[c]$ to denote the command obtained by replacing the hole ($\square$) in K by $c$.

We use $c[v/x]$ to denote substitution of a value $v$ for a variable $x$ in command $c$. Notice that our expressions $e$ never get stuck and have no side-effects. We treat closed expressions that evaluate to the same value as equal.

\begin{figure*}
\begin{mathpar}
\inferrule[Cons]{
\{\ell, \dots, \ell + n - 1\} \cap \mathrm{dom}\,h = \emptyset
}{
(h, A, \mathbf{cons}(v_1, \dots, v_n)) \stackrel{t}{\rightarrow}_\mathsf{h}
(h[\ell:=v_1,\dots,\ell + n - 1:=v_n], A, \ell)
}
\and
\inferrule[NA-Read]{
}{
(h[\ell:=v], A, [\ell]_\mathbf{na}) \stackrel{t}{\rightarrow}_\mathsf{h}
(h[\ell:=v], A, v)
}
\and
\inferrule[NA-Write-Start]{
}{
(h[\ell:=v_0], A, [\ell] :=_\mathbf{na} v) \stackrel{t}{\rightarrow}_\mathsf{h}
(h[\ell:=\varnothing], A, [\ell] :='_\mathbf{na} v)
}
\and
\inferrule[NA-Write-End]{
}{
(h[\ell:=\varnothing], A, [\ell] :='_\mathbf{na} v) \stackrel{t}{\rightarrow}_\mathsf{h}
(h[\ell:=v], A, 0)
}
\and
\inferrule[Free]{
}{
(h[\ell:=v], A, \mathbf{free}(\ell)) \stackrel{t}{\rightarrow}_\mathsf{h}
(h[\ell:=\bot], A, 0)
}
\and
\inferrule[If-True]{
v \neq 0
}{
(h, A, \mathbf{if}\ v\ \mathbf{then}\ c) \stackrel{t}{\rightarrow}_\mathsf{h} (h, A, c)
}
\and
\inferrule[If-False]{
}{
(h, A, \mathbf{if}\ 0\ \mathbf{then}\ c) \stackrel{t}{\rightarrow}_\mathsf{h} (h, A, 0)
}
\and
\inferrule[Let]{}{
(h, A, \mathbf{let}\ x = v\ \mathbf{in}\ c) \stackrel{t}{\rightarrow}_\mathsf{h} (h, A, c[v/x])
}
\and
\inferrule[Head-Step]{
(h, A, c) \stackrel{t}{\rightarrow}_\mathsf{h} (h', A', c')
}{
(h, A, T[t:=K[c]]) \rightarrow (h', A', T[t:=K[c']])
}
\and
\inferrule[Fork]{
t' \neq t\\
t' \notin \mathrm{dom}\,T
}{
(h, A, T[t:=K[\mathbf{fork}(c)]]) \rightarrow (h, A, T[t:=K[0]][t':=c])
}
\end{mathpar}
\caption{Remaining operational semantics step rules}\label{fig:opsem2}
\end{figure*}

The step relation of our operational semantics is defined in Figs.~\ref{fig:opsem} and \ref{fig:opsem2}. The primed nonatomic write
command $[\ell] :='_\mathbf{na} v$ is a syntactic construct that is not allowed to appear in source programs and occurs only
during execution. It denotes a nonatomic write in progress. This way of modeling nonatomic writes, borrowed from RustBelt \cite{jung-phd}, ensures that a program with a data race has a configuration reachable in the opsem where one of the racing threads is stuck.

\subsection{A Verification Approach for C20: Soundness}\label{app:c20-soundness}

In this section, we prove that if a program is safe under the operational semantics, it has no undefined behavior under C20 semantics. We here concentrate on proving data-race-freedom; other types of undefined behavior can be handled similarly.

For the remainder of this section, we fix an instrumented program and we assume that it is safe under the operational semantics.

We fix a C20 execution graph $G$.

In this section we assume that for any location $\ell$ that is accessed atomically, a $\mathbf{begin\_atomic}(\ell, \Sigma)$ event coincides with the initializing write to $\ell$, which happens-before all other accesses of $\ell$, and we assume all accesses of $\ell$ are enabled under $\Sigma$. (We lift these assumptions in \S\ref{sec:yc20-logic}.) We say $\ell$ has atomic specification $\Sigma$.

\begin{definition}[Execution Prefix]
We say a subgraph of the execution is an \emph{execution prefix} if it is prefix-closed with respect to happens-before, i.e. if an event is in the prefix, then all events that happen before it are also in the prefix.
\end{definition}

\begin{definition}[Data race]
A data race is a pair of accesses of the same location, not ordered by happens-before, at least one of which is a write and at least one of which is nonatomic. For the purposes of this definition, we treat allocations, deallocations, and conversion to or from atomic mode as nonatomic writes.
\end{definition}

Notice that under the assumptions of this section, all data races are among nonatomic accesses and atomic accesses are never involved in a data race.

We say a configuration $\gamma$ of the operational semantics corresponds to an execution prefix $P$, denoted $P \sim \gamma$, if all of the following hold:
\begin{itemize}
\item the thread pool matches the final configurations of the threads of the prefix
\item the prefix has no data races.
\item the nonatomic heap maps each allocated cell that is only accessed nonatomically to the value written by its final write.
\item the atomic heap maps each allocated cell $\ell$ that is accessed atomically to its atomic specification and to the tied resource reached by, starting from $(\Sigma.\rho_0, \varepsilon)$, first producing the tied postconditions of the atomic access events on $\ell$ in the prefix, and then consuming their tied preconditions.
\end{itemize}

Importantly, we have $P \sim \gamma \land P \sim \gamma' \Rightarrow \gamma = \gamma'$.

\begin{lemma}
In any reachable configuration of the operational semantics, for each location $\ell$, one of the following holds:
\begin{itemize}
\item Not yet allocated: $h(\ell) = \bot$ and $A(\ell) = \bot$
\item In nonatomic mode: $\exists v.\;h(\ell) = v \lor h(\ell) = \varnothing$ and $A(\ell) = \bot$
\item In atomic mode: $h(\ell) = \varnothing$ and $\exists \Sigma, \omega.\;A(\ell) = (\Sigma, \omega)$
\item Deallocated: $h(\ell) = \varnothing$ and $A(\ell) = \bot$
\end{itemize}
\end{lemma}
\begin{proof}
By induction on the number of steps.
\end{proof}

We say an execution prefix $P$ is \emph{valid} if it corresponds to a configuration $\gamma$ and $\gamma$ is reached by executing the events of $P$ in any order consistent with happens-before.

\subsubsection{Proof of main lemma}

In this sub-subsection, we prove the main lemma, which says that every execution prefix is valid. By induction on the size of the prefix. We fix an execution prefix $P$, and we assume all smaller prefixes are valid. The case where $P$ is empty is trivial; assume $P$ is nonempty.

\begin{lemma}
$P$ is data-race-free.
\end{lemma}
\begin{proof}
By contradiction. Assume there are two events $e, e' \in P$, in threads $t$ and $t'$ not ordered by $\hb$. Obtain $P'$ by removing from $P$ $e$ as well as the events that happen-after $e$ or $e'$. By the induction hypothesis, $P'$ corresponds to some configuration, and this configuration is reachable by executing $e'$ last. Therefore, in this configuration, $t$ is about to execute $e$ and $t'$ has just executed $e'$. By case analysis on $e$ and $e'$, we obtain a contradiction.
\end{proof}

\begin{lemma}
For every order on the events of $P$ consistent with $\hb$, a configuration $\gamma$ such that $P \sim \gamma$ is reachable by executing the events in this order.
\end{lemma}
\begin{proof}
Fix such an order. Let $e$, in thread $t$, be the final event in this order. Apply the induction hypothesis to $P' = P \setminus \{e\}$ to obtain some $\gamma'$ such that $P' \sim \gamma'$. It suffices to prove that executing $e$ in $\gamma'$ reaches a configuration $\gamma$ such that $P \sim \gamma$. By case analysis on $e$. We elaborate one case.
\begin{itemize}
\item \textbf{Case} $e$ is an atomic operation. By the fact that $\gamma'$ is reachable and, by the safety of the program, therefore not stuck, we can, starting from $\gamma'$, consume $e$'s tied precondition and produce its tied postcondition to obtain $\gamma$. It remains to prove that $\Sigma, t, v, o \vDash \omega$, where $\mathsf{lab}(e) = (t, \ell, v, o)$ and $\gamma'.A(\ell) = (\Sigma, \omega)$. Take as the atomic location trace the entire execution graph $G$, with $E_\mathsf{at}$ the set of all atomic operations on $\ell$ in $G$, and $\mathit{init}$ the $\mathbf{begin\_atomic}(\ell, \Sigma)$ event. For $E_\mathsf{ex}$, take the accesses of $\ell$ in $P'$.
\end{itemize}
\end{proof}

\subsection{A Verification Approach for YC20: Soundness}\label{app:yc20-soundness}

In this section, we prove that if a program is safe under the operational semantics, it has no undefined behavior under YC20 \cite{xmm} semantics. We again focus on proving data-race-freedom.

For the remainder of this section, we fix an instrumented program and we assume that it is safe under the operational semantics.

\begin{lemma}
If the program obtained by inserting redundant nonatomic read-write pairs $[e] :=_\mathbf{na} [e]_\mathbf{na}$ into a YC20 program is data-race-free, then the original program is data-race-free as well.
\end{lemma}

We apply YC20 semantics to the instrumented program by treating the $\mathbf{begin\_atomic}$ and $\mathbf{end\_atomic}$ commands like redundant nonatomic read-write pairs.

%

\subsubsection{Grounded C20 executions}

\begin{definition}
We say an atomic access event $e$ on a location $\ell$ is \emph{grounded} with respect to a $\mathbf{begin\_atomic}(\ell, \Sigma)$ event if $e$ as well as $e$'s $\rf^+$ predecessors are enabled under $\Sigma$.
\end{definition}

\begin{definition}
We say an atomic access event $e$ on a location $\ell$ is \emph{grounded} if it has exactly one $\hb$-maximal $\hb$-preceding $\mathbf{begin\_atomic}(\ell, \Sigma)$ event $e'$ and $e$ is grounded with respect to $e'$.
\end{definition}

\begin{definition}
We say a C20 execution is \emph{grounded} if every atomic access event in this execution is grounded.
\end{definition}

\begin{definition}
We say an atomic access event $e$ on a location $\ell$ is \emph{weakly grounded} under some \emph{grounding order} (a total order on the events of the execution) if at least one of the following is true:
\begin{itemize}
\item At least one $\mathbf{begin\_atomic}(\ell, \_)$ event precedes $e$ in grounding order and $e$ is grounded with respect to the most recent preceding one in grounding order
\item both of the following are true:
\begin{itemize}
\item for any $\rf^+$ predecessor $e'$ of any transitive-reflexive grounding-order-predecessor of $e$, if $e'$ is a release event then $e'$ precedes $e$ in the grounding order
\item $e$ is a write or fence event or $e$'s $\rf$ predecessor exists and is before $e$ in the grounding order and is also weakly grounded or is a nonatomic access
\end{itemize}
\end{itemize}
\end{definition}

\begin{definition}
We say a C20 execution is \emph{weakly grounded} if there is a single total order on the events of the execution (called its \emph{grounding order}) consistent with happens-before such that every atomic access event of the execution is weakly grounded.
\end{definition}

\begin{lemma}\label{lem:exec-step-weakly-grounded}
If an Execute step goes from $G$ to $G'$, and $G$ is grounded, then $G'$ is weakly grounded. For the
grounding order, take some arbitrary total order on the events of $G$ consistent with happens-before, followed by the newly added event.
\end{lemma}

\begin{lemma}\label{lem:non-determined-committed-not-release}
During a Re-Execute step, a non-determined committed event $e$ is not a release event.
\end{lemma}
\begin{proof}
By contradiction; assume $e$ is a release event. $e$ must have an uncommitted $\po$-predecessor $e'$; otherwise, $e$ would be determined. But since $e$ is
a release event, there is an $\rpo$ edge from $e'$ to $e$. But only determined events are allowed to have outgoing $\rpo$ edges.
\end{proof}

\begin{lemma}
During a Re-Execute step, the order in which the events are added is consistent with happens-before.
\end{lemma}
\begin{proof}
By contradiction. Assume that a Guided Step that adds an event $e$ also adds a happens-before edge from $e$ to some existing event. Since $e$ is $\po$-maximal, it must be that this edge is a synchronizes-with edge, which implies that $e$ is a release event with a reads-from edge to some existing event. This implies $e$ is a committed event. Since it is not a determined event, we obtain a contradiction by Lemma~\ref{lem:non-determined-committed-not-release}.
\end{proof}

\begin{lemma}\label{lem:re-exec-step-weakly-grounded}
If a Re-Execute step goes from $G$ to $G'$ using committed events $C$, and $G$ is grounded, then $G'$ is weakly grounded. For the
grounding order, take some arbitrary order on the determined events consistent with the happens-before order of $G'$, followed by the other events in the order in which they are added by the Guided Steps.
\end{lemma}
\begin{proof}
We prove, by induction on the number of preceding Guided Steps, that for the event $e$ added by a Guided Step, all $\rf^+$-predecessors that are release events were added earlier. If $e$ is an uncommitted event with an $\rf$-predecessor $e'$, we have that $e'$ was added earlier, so by the induction hypothesis we have the goal. If $e$ is a committed event, its $\rf^+$-predecessors are also committed events, since the committed set is $\rf$-complete. By Lemma~\ref{lem:non-determined-committed-not-release} we have the goal.
\end{proof}

\subsubsection{Proving one XMM step}

We fix a weakly grounded C20 execution graph $G$ and a corresponding grounding order.

\begin{definition}[Execution Prefix]
We say a subgraph of the execution is an \emph{execution prefix} if it is prefix-closed with respect to happens-before, i.e. if an event is in the prefix, then all events that happen before it are also in the prefix.
\end{definition}

\begin{definition}[Data race]
A data race is a pair of accesses of the same location, not ordered by happens-before, at least one of which is a write and at least one of which is nonatomic. For the purposes of this definition, we treat allocations, deallocations, and conversion to or from atomic mode as nonatomic writes.
\end{definition}

If an execution prefix is data-race-free, then it is well-defined at each access of a location within that prefix whether at that event the location is in nonatomic mode or in atomic mode, and, if it is in atomic mode, what its atomic specification is, based on the $\hb$-maximal conversion event that happens before the access. We say the execution prefix is mode-well-formed if nonatomic accesses occur only on locations in nonatomic mode, and atomic accesses occur only on locations in atomic mode.

We say a configuration $\gamma$ of the operational semantics corresponds to an execution prefix $P$, denoted $P \sim \gamma$, if all of the following hold:
\begin{itemize}
\item the thread pool matches the final configurations of the threads of the prefix
\item the prefix has no data races. It follows that each location's final mode (atomic or nonatomic) and (in the case of nonatomic locations) final value within the prefix is well-defined.
\item the prefix is mode-well-formed
\item the prefix is grounded.
\item the nonatomic heap maps each allocated cell whose final mode is nonatomic to the value written by its final write.
\item the atomic heap maps each allocated cell $\ell$ whose final mode is atomic with atomic specification $\Sigma$ to the atomic specification assigned to it by its latest (i.e.~$\hb$-maximal) $\mathbf{begin\_atomic}$ event $e$, and to the tied resource obtained by, starting from $(\Sigma.\rho_0, \varepsilon)$, first producing the tied postconditions and then consuming the tied preconditions of the atomic access events on $\ell$ in the prefix that happen-after $e$.
\end{itemize}

Notice that the configuration corresponding to an execution prefix $P$, if it exists, is unique. We denote it by $\gamma_P$.

\begin{lemma}
In any reachable configuration of the operational semantics, for each location $\ell$, one of the following holds:
\begin{itemize}
\item Not yet allocated: $h(\ell) = \bot$ and $A(\ell) = \bot$
\item In nonatomic mode: $\exists v.\;h(\ell) = v \lor h(\ell) = \varnothing$ and $A(\ell) = \bot$
\item In atomic mode: $h(\ell) = \varnothing$ and $\exists \Sigma, \omega.\;A(\ell) = (\Sigma, \omega)$
\item Deallocated: $h(\ell) = \varnothing$ and $A(\ell) = \bot$
\end{itemize}
\end{lemma}
\begin{proof}
By induction on the number of steps.
\end{proof}

We say an execution prefix $P$ is \emph{valid} if it corresponds to a configuration $\gamma$ and $\gamma$ is reached by executing the events of $P$ in any order consistent with happens-before.

We define $P_n$ as the prefix constituted by the first $n$ events in grounding order.

\subsubsection{Proof of main lemma}

In this sub-subsection, we prove the main lemma, which says that, for all $n$, all sub-prefixes of $P_n$ are valid. By induction on $n$. We fix an $n$ and we assume all sub-prefixes of all $P_m$ with $m < n$ are valid. The case where $n = 0$ is trivial; assume $n > 0$.

In the following two lemmas, let $e$ be the grounding-order-maximal event in $P_n$, and let $t$ be the thread of $e$.

\begin{lemma}
$P_n$ is data-race-free.
\end{lemma}
\begin{proof}
By contradiction. By the induction hypothesis, we have $P_{n - 1}$ is data-race-free, so there must be some event $e' \in P_{n-1}$ in some thread $t'$
that races with $e$. Consider the prefix $P$ obtained by removing from $P_{n-1}$ all events that happen-after $e'$. Since $P$ is valid,
we have a $\gamma$ such that $P \sim \gamma$ and $\gamma$ is reachable by executing the events of $P$ in any order. Assume $e'$ was executed last. So in $\gamma$, $t'$ has just executed $e'$ and $t$ is about to execute $e$. By case analysis on $e$ and $e'$.
\end{proof}

\begin{lemma}
$P_n$ is grounded.
\end{lemma}
\begin{proof}
By the induction hypothesis, we have that $P_{n - 1}$ is grounded. It remains to prove that $e$ is grounded. Assume it is weakly grounded. By $\gamma_{P_{n-1}}$ not being stuck, we have that $e$'s operation $o$ is enabled and that the current mode of $e$'s location $\ell$ is atomic, with some atomic specification $\Sigma$.
\begin{itemize}
\item We prove that $e$'s $\rf^+$ predecessors are enabled under $\Sigma$. $e$'s immediate $\rf$-predecessor $e'$, if any, is in $P_{n-1}$ and is therefore grounded; by groundedness of $e'$ the goal follows for $e$'s indirect $\rf^+$-predecessors.
\item We exploit sufficiency of $e$'s tied precondition. Let $E_\mathsf{ex}$ be the set of all atomic access events on $\ell$ in $P_{n - 1}$ that happen-after the $\hb$-maximal $\mathbf{begin\_atomic}$ event on $\ell$ in $P_n$. We construct an atomic location trace with $E_\mathsf{at}$ consisting of $e$, all $\rf^+$ predecessors of $e$, and all $\rf^*$ predecessors of the events in $E_\mathsf{ex}$. By sufficiency, we have that $e$ is enabled in $\Sigma$.
\end{itemize}
\end{proof}

\begin{lemma}\label{lem:yc20-main}
For every sub-prefix $P$ of $P_n$, and for every order on the events of $P$ consistent with $\hb$, a configuration $\gamma$ such that $P \sim \gamma$ is reachable by executing the events in this order.
\end{lemma}
\begin{proof}
By induction on the size of $P$. Fix a $P$ and fix such an order. Let $e$, in thread $t$, be the final event in this order. Apply the induction hypothesis to $P' = P \setminus \{e\}$ to obtain some $\gamma'$ such that $P' \sim \gamma'$. It suffices to prove that executing $e$ in $\gamma'$ reaches a configuration $\gamma$ such that $P \sim \gamma$. By case analysis on $e$. We elaborate one case.
\begin{itemize}
\item \textbf{Case} $e$ is an atomic operation. Assume $\mathsf{lab}(e) = (t, \ell, v, o)$. By the fact that $\gamma'$ is reachable and, by the safety of the program, therefore not stuck, we know $\gamma'.A(\ell) = (\Sigma, \omega)$ for some $\Sigma$ and $\omega$, and we can, starting from $\gamma'$, consume $e$'s tied precondition and produce its tied postcondition to obtain $\gamma$. It remains to prove that $\Sigma, t, v, o \vDash \omega$. Let $E_\mathsf{ex}$ be the set of atomic accesses on $\ell$ in $P'$ that happen-after the $\hb$-maximal $\mathbf{begin\_atomic}(\ell, \_)$ event $e'$ in $P'$. Take as the set $E_\mathsf{at}$ of the atomic location trace the $\rf^{-1}$ closure of $E_\mathsf{ex} \cup \{e\}$, i.e.~the events of $E_\mathsf{ex} \cup \{e\}$ as well as all $\rf^+$-predecessors of those events.
\end{itemize}
\end{proof}

\subsubsection{Proving an YMM trace}

\begin{theorem}
The program is grounded and data-race-free under YC20.
\end{theorem}
\begin{proof}
By induction on the number of YMM steps, i.e.~the number of (Re-)Execute steps. Base case: the empty execution is grounded and data-race-free. Induction step: assume $G$ is grounded and data-race-free, and assume the step goes from $G$ to $G'$. From Lemmas~\ref{lem:exec-step-weakly-grounded} and \ref{lem:re-exec-step-weakly-grounded} we know $G'$ is weakly grounded. By Lemma~\ref{lem:yc20-main}, we obtain that $G'$ is grounded and data-race-free.
\end{proof}

\end{document}